\documentclass[10pt,conference]{IEEEtran}
\usepackage[cmex10]{amsmath}
\usepackage{amssymb}
\usepackage{amsfonts}

\newtheorem{definition}{Definition} 
\newtheorem{proposition}{Proposition}
\newtheorem{lemma}{Lemma}

\newcommand{\pnt}[1]{{\mbox{\boldmath $#1$}}}
\newcommand{\cof}[2]{\mbox{$#1_{\boldsymbol{#2}}$}}

\newcommand{\nGz}[2]{$G_{non-\{z\}}$}

\newcommand{\mi}[1]{\mathit{#1}}
\newcommand{\ti}[1]{\textit{#1}}
\newcommand{\tb}[1]{\textbf{#1}}

\newcommand{\ttt}{\>\>\>}

\newcommand{\Tt}{\>\>}

\newcommand{\SUb}[2]{\mbox{$\mi{#1}_\mi{#2}$}}
\newcommand{\Spb}[3]{\mbox{$#1^\mi{#2}_\mi{#3}$}}
\newcommand{\SPb}[3]{\mbox{$#1^\mi{#2}_{#3}$}}
\newcommand{\Sup}[2]{\mbox{$#1^\mi{#2}$}}

\newcommand{\prob}[2]{\mbox{$\exists{#1} [#2]$}}
\newcommand{\Prob}[3]{\mbox{$\exists{#1}\exists{#2}[#3]$}}

\newcommand{\Comment}[1]{}

\newcommand{\nimpl}[2]{\mbox{$#1 \not\rightarrow #2$}}

\newcommand{\tr}[3]{\mbox{(\pnt{#1_{#2}},\dots,\pnt{#1_{#3}})}}

\newcommand{\impl}[2]{\mbox{$#1 \rightarrow #2$}}
\newcommand{\ks}{\mbox{$\xi$}}
\newcommand{\Ks}{\mbox{$\xi$}~}
\newcommand{\ksr}{\Sup{\ks}{rlx}}
\newcommand{\Ksr}{\Sup{\ks}{rlx}~}
\newcommand{\ksrr}[1]{\Spb{\ks}{rlx}{#1}}

\newcommand{\Et}{\mbox{$\eta$}~}

\newcommand{\alg}{\ti{PC\!\_LoR}}
\newcommand{\Alg}{\ti{PC\!\_LoR}~}

\newcommand{\Mix}{\ti{LoR\_IC}~}

\newcommand{\abs}[1]{\mbox{$\mathbb{#1}$}}
\newcommand{\Abs}[2]{\mbox{$\mathbb{#1}_{#2}$}}

\mathchardef\mhyphen="2D

\newcommand{\Fii}{\mbox{$\phi$}~}

\newcommand{\trx}[1]{\mbox{$\SPb{T}{rlx}{#1}$}}
\newcommand{\ttrx}[1]{\mbox{$\SPb{\abs{T}}{rlx}{#1}$}}
\newcommand{\Ttrx}[1]{\mbox{$\SPb{\abs{T}}{RLX}{#1}$}}

\usepackage{graphicx}
\usepackage{enumerate}
\begin{document}

\title{Property Checking By Logic Relaxation}

\author{\IEEEauthorblockN{Eugene Goldberg} 
\IEEEauthorblockA{
eu.goldberg@gmail.com}}

\maketitle

\begin{abstract}
We introduce a new framework for Property Checking (PC) of sequential
circuits. It is based on a method called \tb{Lo}gic \tb{R}elaxation
(LoR). Given a safety property, the LoR method relaxes the transition
system at hand, which leads to expanding the set of reachable states.
For $j$-th time frame, the LoR method computes a superset $A_j$ of the
set of bad states reachable in $j$ transitions \ti{only} by the
relaxed system.  Set $A_j$ is constructed by a technique called
partial quantifier elimination.  If $A_j$ does not contain a bad state
and this state is reachable in $j$ transitions in \ti{the relaxed
  system}, it is also reachable in the original system. Hence the
property in question does not hold.

The appeal of PC by LoR is as follows. An inductive invariant (or a
counterexample) generated by LoR is a result of computing the states
reachable only in the relaxed system.  So, the complexity of PC can be
drastically reduced by finding a ``faulty'' relaxation that is close
to the original system. This is analogous to equivalence checking
whose complexity strongly depends on how similar the designs to be
compared are.
\end{abstract}

\section{Introduction}
%
%
\subsection{Motivation}
\label{subsec:motivation}
Property checking is an important part of the formal verification of
hardware. Recently, new powerful methods of property checking have
been developed~\cite{interpolation,ic3}. A characteristic feature of
those methods is that they use SAT-solving to avoid operating on
quantified formulas e.g. performing quantifier elimination. This is
done because of insufficient efficiency of the current algorithms for
quantified logic. On the other hand, such algorithms have great
potential because reasoning on quantified formulas facilitates very
powerful transformations preserving equi-satisfiability rather than
equivalence.

To address the problem of reasoning on quantified formulas, we have
been developing a machinery of Dependency sequents
(D-sequents)~\cite{fmcad12,fmcad13,fmsd14}. In particular, we have
introduced a technique called Partial Quantifier Elimination
(PQE)~\cite{hvc-14} that can boost the performance of algorithms
operating on quantified formulas. Our research on D-sequents and PQE
is still work in progress and we believe that catching up with
SAT-based algorithms is just a matter of time.  So we try to combine
work on improving PQE algorithms with research that explains the
benefits of such algorithms for formal
verification~\cite{tech_rep_crr,tech_rep_ec_sim,tech_rep_ec_lor}.  In
particular, in ~\cite{tech_rep_ec_lor}, we introduced a new
verification method called Logic Relaxation (LoR) enabled by PQE.  We
showed that applying the LoR method to equivalence checking of
combinational circuits facilitates generation of powerful inductive
proofs.  In this paper, we continue this work by applying the LoR
method to property checking of sequential circuits.

%
%
\subsection{Problem formulation}
\label{subsec:problem}
Let $M(S,X,Y,Z,S')$ be a sequential circuit where $X$, $Y$ and $Z$ are
sets of input, internal and output combinational variables
respectively, $S$ and $S'$ are sets of present and next state
variables respectively.  Let $T(S,X,Y,S')$ be a formula representing
the transition relation specified by $M$.  All formulas we consider in
this paper are Boolean. We will assume that every formula is
represented in the Conjunctive Normal Form (CNF).  We will call a
complete assignment \pnt{s} to state variables a \pnt{state}.
Henceforth, by an assignment \pnt{v} to a set of variables $V$ we mean
a \ti{complete} assignment unless otherwise stated.  Denote by \Ks the
transition system specified by transition relation $T$ and a set of
initial states $I(S)$.  Let $P(S)$ specify the property of \Ks to be
verified.  Given a Boolean formula $A(S)$, a state \pnt{s} is called
an $A$-state if $A(\pnt{s})=1$. We will refer to $P$-states and
$\overline{P}$-states as \tb{good} and \tb{bad} ones respectively.  In
this paper, we consider checking a safety property. That is, given a
property $P$, one needs to prove that a) no bad state of \Ks is
reachable from an $I$-state or b) a counterexample exists.

%
%
\subsection{Property checking by logic relaxation}
\label{subsec:pc_by_lor}
Denote by \Ksr a ``relaxed'' version of system \Ks. Both \Ks and \Ksr
have the same set of initial states but are different in their
transition relations.  Let $S_j,X_j,Y_j$ denote sets of variables of
\Ks in $j$-th time frame. Let $T_{j,j+1}$ denote formula
$T(S_j,X_j,Y_j,S_{j+1})$ i.e. the transition relation of \Ks in $j$-th
time frame.  Let \trx{j,j+1} denote formula
$\Sup{T}{rlx}(S_j,X_j,Y_j,S_{j+1})$ specifying the transition relation
of \Ksr in $j$-th time frame.  Formula $T_{j,j+1}$ implies
\trx{j,j+1}, so the set of transitions allowed in \Ksr is a superset
of those in \ks.

The idea of Property Checking (PC) by Logic Relaxation (LoR) is as
follows.  Since the set of valid traces of \Ksr is a superset of
that of \ks, a state reachable in $j$-th time frame of \Ks is also
reachable in \ksr.  Suppose that one has computed a set containing
all states reachable in $j$-th time frame \ti{only} in \Ksr.  Then
the existence of a bad state \pnt{s} that is not in this set and is
reachable in the \ti{relaxed} system \Ksr means that \pnt{s} is
reachable in \Ks as well and property $P$ fails.

%
%
\subsection{Boundary formulas}
\label{subsec:bnd_forms}
A key part of PC by LoR is computing so-called boundary formulas
computing supersets of states reachable only in \ksr.  Formula
$H_j(S_j)$ is called \tb{boundary} for the pair (\Ks, \Ksr) if it
\begin{itemize}
\item evaluates to 0 for every state that
is reachable in \Ksr but not in \Ks in $j$ transitions
\item evaluates to 1 for every state that is reachable in system \Ks
  (and hence in \Ksr) in $j$ transitions
\end{itemize}
The value of $H_j$ is not specified for a state that is unreachable in
\Ksr (and hence in \ks) in $j$ transitions. On the one hand, $H_j$ can
be viewed as a ``boundary'' between sets of states reachable in \Ks
and \Ksr in $j$ transitions, hence the name. On the other hand, since
every $\overline{H_j}$-state is unreachable in \Ks in $j$ transitions,
the $H_j$-states form an over-approximation of the set of states
reachable in \Ks in $j$ transitions.

\subsection{Transition relation relaxation}
Let us show how one can use transition relation relaxation to build
boundary formula $H_1$. The latter gives an over-approximation of the
set of states reachable in \Ks in one transition. Suppose that no bad
state is reachable from an $I$-state of \Ks in one transition.  Let
\pnt{s} be a bad state. Since \pnt{s} is unreachable from an $I$-state
in one transition, formula $I_0 \wedge T_{0,1} \wedge \cof{C}{s}$ is
unsatisfiable. Here $I_0$ denotes $I(S_0)$ and \cof{C}{s} is the
longest clause falsified by \pnt{s}. (A \tb{clause} is a disjunction
of literals).  Let us relax $T_{0,1}$ to make state \pnt{s}
reachable. This means finding a formula \trx{0,1} implied by $T_{0,1}$
that makes $I_0 \wedge \trx{0,1} \wedge \cof{C}{s}$ satisfiable. Let
$R_{0,1}$ be a formula such that $T_{0,1} \equiv \trx{0,1} \wedge
R_{0,1}$ i.e. $R_{0,1}$ specifies the ``difference'' between the
transition relations. In the simplest case, $R_{0,1}$ is just a subset
of clauses of $T_{0,1}$ and so \trx{0,1} is obtained from $T_{0,1}$ by
removing the clauses of $R_{0,1}$.  (In this paper, we use the notion
of a CNF formula and that of a set of clauses interchangeably.)

Boundary formula $H_1$ is built by excluding states reachable only by
the relaxed system \ksr, specified by \trx{0,1}, in one
transition. Initially, $H_1$ is an empty set of clauses that
represents a constant 1.  Let $G(S_1)$ be a formula such that
$\prob{W_0}{I_0 \wedge \trx{0,1} \wedge R_{0,1}} \equiv G \wedge
\prob{W_0}{I_0 \wedge \trx{0,1}}$ where $W_0 = X_0 \wedge Y_0 \wedge
S_0$. Finding $G$ comes down to solving the Partial Quantifier
Elimination (PQE) problem. (Only a part of the quantified formula
leaves the scope of quantifiers, hence the name.)  States falsifying
$G$ is a superset of states reachable with transition relation
\trx{0,1} but not with $T_{0,1}$. In particular, $G$ is falsified by
state \pnt{s}. The clauses of $G$ are added to $H_1$. If \impl{H_1}{P}
holds, then $H_1$ is an over-approximation of the set of states
reachable in \Ks in one transition. Otherwise, there is a bad state
\pnt{s} for which formula $I_0 \wedge \trx{0,1} \wedge H_1 \wedge
\cof{C}{s}$ is unsatisfiable. Then the procedure above is applied
again. That is transition relation \trx{0,1} is relaxed even more and
a new formula $G$ falsified by \pnt{s} is derived that makes up for
this new relaxation. The set of clauses of $G$ is added to $H_1$. This
goes on until \impl{H_1}{P} holds.

\subsection{A high-level view of \Alg}
In this paper, we formulate a an algorithm of PC by LoR called \alg.
To check if a property $P$ holds, \Alg computes a sequence of boundary
formulas $H_1,\dots,H_j$ that satisfy properties similar to those
maintained in IC3~\cite{ic3}.  However these formulas are derived by
employing transition relation relaxation and PQE rather than inductive
clauses. Maintaining IC3-like properties is just a convenient way to
guarantee that \Alg converges. If $P$ holds in \ks, then eventually
logically equivalent boundary formulas $H_j$ and $H_{j+1}$ are
produced, meaning that $H_j$ is an inductive invariant. Otherwise,
\Alg fails to build a boundary formula $H_j$ implying property $P$ and
finds a counterexample instead.

We also describe a version of \Alg that combines LoR with derivation
of inductive clauses~\cite{ic3}. In IC3, a formula $F_j$
over-approximating the set of states reachable in $j$ transitions is
built by tightening $P$. This tightening is done by adding to $F_j$
inductive clauses excluding $F_j$-states from which a bad state is
reachable in one transition. Such an approach may converge too slowly
if an inductive invariant is ``far'' from property $P$. The idea of
combining LoR with derivation of inductive clauses is as follows. The
original boundary formula $H_i$ is built by relaxation.  The future
corrections of $H_i$ (done to maintain the IC3-like properties we
mentioned above) are performed by tightening $H_i$ up by inductive
clauses.  Such an approach can drastically speed up building an
inductive invariant that is far from the property.

%
\subsection{Merits of PC by LoR}

This paper is motivated by some nice features of PC by LoR listed
below.  Since PC by LoR heavily relies on existence of efficient PQE
solvers, realization of these features requires a boost in the
performance of current PQE algorithms. We believe that this can be
achieved via implementing some crucial
techniques~\cite{tech_rep_ec_lor} that PQE solvers still lack. So
getting the required performance of PQE is just a matter of time.

Our interest in PC by LoR is twofold.  First, PC by LoR derives an
inductive invariant (or a counterexample) by computing the difference
between the original and relaxed transition systems. So, in a sense,
the complexity of PC becomes \ti{relative} since it depends on how
different the original and relaxed systems are.  This is analogous to
equivalence checking whose complexity strongly depends on how similar
the designs to be compared are. Second, by using a particular
relaxation scheme one can take into account system and property
structure/semantics. Suppose, for instance, that one needs to check a
property $P$ of a system \Ks induced by interaction of two its
subsystems $\ks'$ and $\ks''$. Intuitively, an inductive invariant can
be constructed by computing the difference between \Ks and a relaxed
system obtained from \Ks by removing the interaction between $\ks'$
and $\ks''$. If $P$ holds, then bad states are reachable only in the
relaxed system. That is the knowledge of problem semantics may help to
generate an inductive invariant faster.  We show how this idea works
for equivalence checking (Section~\ref{sec:exmp}).

%
%
\subsection{Contributions and structure of the paper}
The contribution of this paper is threefold. First, we introduce a new
framework for PC. It is based on the idea of using transition relation
relaxation and PQE to build an over-approximation of the set of
reachable states.  Second, we formulate a PC algorithm based on this
idea and prove its correctness. Third, we formulate a PC algorithm
combining transition relation relaxation with the machinery of
inductive clauses.

The remainder of the paper is structured as follows. An example of PC
by LoR is described in Section~\ref{sec:exmp}. Basic definitions are
given in Section~\ref{sec:definitions}.  Boundary formulas are
discussed in Section~\ref{sec:bnd_form}.  We describe \Alg in
Section~\ref{sec:alg_descr}. Section ~\ref{sec:two_mods} discusses two
important modifications of \alg. One of these modifications describes
combining LoR with the machinery of inductive clauses. Some
conclusions are given in Section~\ref{sec:conclusions}.

\section{An Example}
\label{sec:exmp}
In this section, we consider a special case of PC: equivalence
checking of two identical sequential circuits. In
Subsection~\ref{subsec:exmp_descr}, we describe the example we
consider. The problems with solving this example by interpolation and
IC3 are discussed in Subsection~\ref{subsec:ec_by_ii}.  Application of
PC by LoR to this example is described in
Subsection~\ref{subsec:ec_by_lor}.  In particular, such application
shows that by picking transition relation relaxation one can tailor a
PC algorithm to the problem at hand.

%
%
\subsection{Example description}
\label{subsec:exmp_descr}
Let $T^N(X^N,Y^N,S^N,S'^N)$ be a formula specifying the transition
relation of a sequential circuit $N$. Here $X^N,Y^N$ are the sets of
input and internal variables of $N$ respectively and $S^N, S'^N$ are
the sets of present and next state variables of $N$ respectively.  Let
$I^N$ be a formula specifying the initial states of $N$. Let circuit
$K$ be an identical copy of $N$. Let $T^K(X^K,Y^K,S^K,S'^K)$ and $I^K$
be formulas specifying the transition relation and initial states of
$K$. Suppose that one needs to verify equivalence of $K$ and $N$
defined as follows. $K$ and $N$ produce the same sequence of outputs
for an identical sequence of values of $X^N$ and $X^K$ if they start
in the same $I$-state.

The equivalence of $N$ and $K$ can be checked via building a
sequential circuit $M$ called a \tb{miter} that is composed of $N$ and
$K$ as shown in Figure~\ref{fig:sec}. Let $T(X,Y,S,S')=T^N \wedge T^K
\wedge \mi{EQ}(X^N,X^K)$ where $X = X^N \cup X^K$, $Y = Y^N \cup Y^K$,
$S = S^N \cup S^K$, $S' = S'^N \cup S'^K$. Given assignments \pnt{x^N}
and \pnt{x^K} to $X^K$ and $X^N$ respectively,
$\mi{EQ}(\pnt{x^N},\pnt{x^K})=1$ iff \pnt{x^N} = \pnt{x^K} . Formula
$T$ specifies the transition relation of miter $M$.  Formula
$I(S^N,S^K)$ specify the initial states of miter $M$ where
$I(\pnt{s^N},\pnt{s^K})=1$ iff \pnt{s^N} = \pnt{s^K} and
$I^N(\pnt{s^N})=I^K(\pnt{s^K})=1$.  Note that the output variable $z$
of $M$ evaluates to 1 in \mbox{$j$-th} time frame iff $N$ and $K$
produce different assignments to output variables $Z^N$ and $Z^K$.  So
proving the equivalence of $N$ and $K$ comes down to showing that the
output $z$ of miter $M$ evaluates to 0 in every time frame.  This can
be done by proving that the following property $P(S^N,S^K)$ of $M$
holds.  $P(\pnt{s^N},\pnt{s^K})=1$ iff $N$ and $K$ produce the same
outputs in states \pnt{s^N} and \pnt{s^K} for every assignment
\pnt{x^N} to $X^N$ and \pnt{x^K} to $X^K$ such that \pnt{x^N} =
\pnt{x^K}.

Since circuits $N$ and $K$ are identical, the output $z$ of $M$
evaluates to 0 for every state \pnt{s} = (\pnt{s^N},\pnt{s^K}) where
\pnt{s^N} = \pnt{s^K}. So \impl{\mi{EQ(S^N,S^K)}}{P}. However, in
general, the reverse implication does not hold because $N$ and $K$ can
produce the same output even in a state \pnt{s} where $\pnt{s^N} \neq
\pnt{s^K}$. Note that $\mi{EQ}(S^N,S^K) \wedge T \rightarrow
\mi{EQ}(S'^N,S'^K)$ holds. So $\mi{EQ}(S^N,S^K)$ is an inductive
invariant.

\setlength{\intextsep}{4pt}
\begin{figure} 
 \begin{center}
    \includegraphics[width=2.6in]{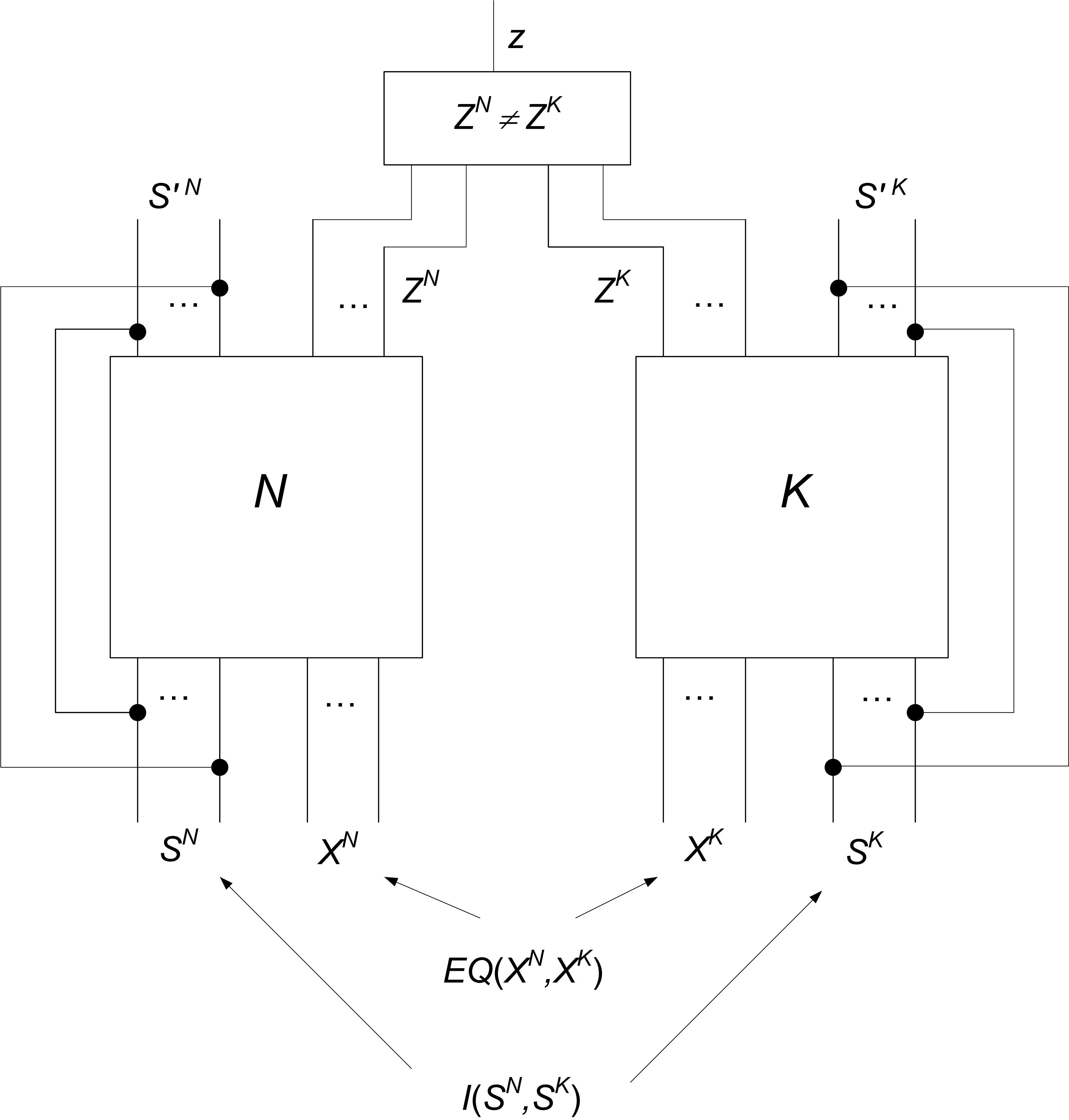}
  \end{center}
\vspace{-10pt}
\caption{Miter $M$ of sequential circuits $N$ and $K$}
\vspace{5pt}
\label{fig:sec}
\end{figure}

%
%
\subsection{Solving example by interpolation and IC3}
\label{subsec:ec_by_ii}

Our example can be trivially solved by a method that tries to prove
equivalence of corresponding state variables of $N$ and $K$. However,
such a method is unrobust since it can not be extended to the case
where $N$ and $K$ are structurally close but not identical (e.g.  if
$N$ does not have state variables that are functionally equivalent to
variables of $K$.) So it is interesting to analyze solving our example
by a general method that does not use pre-processing to identify
equivalent state variables.

One can argue that checking the equivalence of two identical circuits
can be hard for an interpolation based method.  The performance of
such a method strongly depends on the quality of an interpolant
extracted from a proof that $P$ holds for a limited number of
transitions. Such a proof is produced by a general-purpose SAT-solver
based on conflict clause learning. A known fact is that such solvers
generate proofs of poor quality on equivalence checking
formulas~\cite{sat09,tech_rep_ec_lor}. This leads to producing
interpolants of poor quality and hence slow convergence.

IC3 builds an inductive invariant by tightening property $P$ via
adding inductive clauses. Intuitively, the convergence rate of such a
strategy strongly depends on how ``far'' an inductive invariant is
from $P$. Consider, for instance, the inductive invariant
$\mi{EQ}(S^N,S^K)$. In general, $\mi{EQ}$ can be arbitrarily far from
$P$ especially if the transition system specified by $N$ is deep and
$N$ can produce the same output in different states.

%
%
\subsection{Solving example by LoR}
\label{subsec:ec_by_lor}
Let us consider how our example is solved by LoR. As we mentioned in
Subsection~\ref{subsec:pc_by_lor}, the basic operation of PC by LoR is
to compute a superset of the set of states reachable in $j$
transitions only by the relaxed system.  Importantly, this superset is
different from the precise set of reachable states only by the states
(bad or good) that are \ti{unreachable} by the relaxed system and
hence by the original system. The objective here is to make sure that
this superset contains all the bad states. Let us show how this is
done for our example for the initial time frame.  Let \Ks denote the
transition system specified by miter $M$ and initial set of states
$I$.  Recall that transition relation $T_{0,1}$ of \Ks is specified by
$T^N_{0,1} \wedge T^K_{0,1} \wedge \mi{EQ}^X_0$ where $\mi{EQ}^X_0$
denote $\mi{EQ}(X^N_0,X^K_0)$.

Let the relaxed transition \trx{0,1} for the initial time frame be
obtained by dropping the clauses of $\mi{EQ}^X_0$ i.e.  \trx{0,1} =
$T^N_{0,1} \wedge T^K_{0,1}$.  Let \Ksr denote the version of \Ks
where $T_{0,1}$ is replaced with \trx{0,1}.  In \ksr, one can apply
different input assignments to $N$ and $K$. So \Ksr can transition to
states \pnt{s}=(\pnt{s^N},\pnt{s^K}) where $\pnt{s^N} \neq \pnt{s^K}$
and thus potentially reach bad states in one transition. Let us
compute a boundary formula $H_1$. As we mentioned in
Subsection~\ref{subsec:bnd_forms},
\begin{itemize}
\item the $\overline{H_1}$-states specify a superset of the set of
states reachable in one transition only in \Ksr and
\item the $H_1$-states is an over-approximation of the set of states
  reachable in \Ks in one transition.
\end{itemize}
As we show in Section~\ref{sec:bnd_form}, $H_1$ can be found as a
formula for which $\prob{W_0}{I_0 \wedge T_{0,1}} \equiv H_1 \wedge
\prob{W_0}{I_0 \wedge \trx{0,1}}$ holds where $W_0 = X_0 \wedge Y_0
\wedge S_0$.  From Lemma~\ref{lemma:ec} proved in the appendix it
follows that formula $H_1$ equal to $\mi{EQ}(S^N_1,S^K_1)$ satisfies
the equality above. That is $\mi{EQ}(S^N_1,S^K_1)$ can be used as a
boundary formula $H_1$.

As we mentioned earlier, formula $\mi{EQ}(S^K,S^N)$ is an inductive
invariant. So by using a relaxed transition relation \trx{0,1} and
building a boundary formula $H_1$ separating \Ks and \Ksr one
generates an inductive invariant. Such fast convergence is not a
result of pure luck. The choice of relaxation above has a very simple
explanation. Miter $M$ consists of circuits $N$ and $K$
``interacting'' with each other via combinational input
variables. Circuits $N$ and $K$ interact correctly if the output of
$M$ is always 0. Intuitively, to verify that $N$ and $K$ interact
correctly one needs to compute the difference between the original
miter and a relaxed one where communication between $N$ and $K$ is cut
off. If $N$ and $K$ are equivalent, miter $M$ can produce output 1
only if $N$ and $K$ do not talk with each other. In
Subsection~\ref{subsec:lor_ic3}, we argue that such relaxation can be
successfully used in a general algorithm of sequential equivalence
checking.

\section{Basic Definitions}
\label{sec:definitions}
%
%
\begin{definition}
Let \Ks be a transition system specified by transition relation
$T(S,X,Y,S')$ introduced in Subsection~\ref{subsec:problem}. A
sequence of states \tr{s}{m}{j} is called a \tb{trace}.  This trace is
called \tb{valid} if \Prob{X}{Y}{T(\pnt{s_k},X,Y,\pnt{s_{k+1}})} = 1,
$k=m,\dots$,\mbox{$j-1$}.
\end{definition}
%
%
\begin{definition}
Let $I$ specify the initial states of system \ks. Given a property $P$
of \ks, a valid trace \tr{s}{0}{j} is called a \tb{counterexample} if
$I(\pnt{s_0}) = 1$,
$P(\pnt{s_k})=1$,$k=0,\dots$,\mbox{$j-1$},$P(\pnt{s_j})=0$.
\end{definition}
%
%
\begin{definition}
\label{def:relax}
Let $\Ks$ and $\Et$ be two transition systems depending on the same
set of variables $S,X,Y,S'$. We will say that $\Et$ is a
\tb{relaxation} of system $\Ks$ if the set of valid traces of the
former is a superset of that of the latter.
\end{definition}
%
%
\begin{definition}
\label{def:tr_relax}
Let \Ks be a system specified by transition relation $T$ and formula
$I$ specifying initial states. Denote by {\boldmath \ksrr{j}} a
relaxation of \Ks such that
\begin{itemize}
\item \Ks and  \ksrr{j} have identical sets of initial states and
\item $T_{k,k+1} \equiv \trx{k,k+1}$, $k \neq j$ and
\impl{T_{j,j+1}}{\trx{j,j+1}}
\end{itemize}
\end{definition}
%
%

In this paper, by a quantified formula we mean one with
\ti{existential} quantifiers.  Given a quantified formula
\prob{W}{A(V,W)}, the problem of quantifier elimination is to find a
quantifier-free formula $A^*(V)$ such that $A^* \equiv \prob{W}{A}$.
Given a quantified formula \prob{W}{A(V,W) \wedge B(V,W)}, the problem
of \tb{Partial Quantifier Elimination} (\tb{PQE}) is to find a
quantifier-free formula $A^*(V)$ such that $A^* \wedge \prob{W}{B}
\equiv \prob{W}{A \wedge B}$.  Note that formula $B$ remains
quantified (hence the name \ti{partial} quantifier elimination). We
will say that formula $A^*$ is obtained by \tb{taking} \pnt{A} \tb{out
  of the scope of quantifiers} in \prob{W}{A \wedge B}. Importantly,
there is a strong relation between PQE and the notion of
\ti{redundancy} of a clause in a quantified formula. For instance,
solving the PQE problem above comes down to finding a set of clauses
$A^*(V)$ implied by $A \wedge B$ that makes the clauses of $A$
redundant in $A^* \wedge \prob{W}{A \wedge B}$. That is $A^* \wedge
\prob{W}{A \wedge B} \equiv A^* \wedge \prob{W}{B}$.

\section{Boundary Formulas}
\label{sec:bnd_form}
In this section, we present boundary formulas. In
Subsection~\ref{subsec:bf_intro} we define boundary formulas and
explain their relation to PQE. Building boundary formulas inductively
is described in Subsection~\ref{subsec:bf_induct}.
%
%
\subsection{Definition of boundary formulas and their relation to PQE}
\label{subsec:bf_intro}
%
%
\begin{definition}
\label{def:bnd_form}
Let \Ksr be a relaxation of system \Ks and $P$ be a property of
\ks. Formula $H_j$ is called \tb{boundary} for the pair (\ks, \ksr) if
\begin{enumerate}
\item $H_j(\pnt{s})=0$, for every state \pnt{s} that is reachable in
  \Ksr and unreachable in \Ks in $j$ transitions
\item $H_j(\pnt{s})=1$, for every state \pnt{s} that is reachable in
  \Ks (and hence in \ksr) in $j$ transitions
\end{enumerate}
\end{definition}
 Boundary formula $H_j$ specifies the set of states reachable only by
 \Ksr i.e. separates \Ks and \Ksr (hence the name ``boundary'').  We
 will say that $H_j$ is just a boundary formula if the corresponding
 relaxation is obvious from the context.

%
%
Proposition~\ref{prop:bnd_form} below gives a sufficient condition for
a formula to be boundary.  Let system \ksrr{j} be obtained by relaxing
only the transition relation of $j$-th time frame (see
Definition~\ref{def:tr_relax}).  Let $I_0$ denote $I(S_0)$. Let
\Abs{W}{j-1} denote $W_0 \cup \dots \cup W_{j-1}$ where $W_i = S_i
\cup X_i \cup Y_i$, $i=0,\dots,j-1$. Let \Abs{T}{j} denote $T_{0,1}
\wedge \dots \wedge T_{j-1,j}$.  Let \ttrx{j} denote $\Abs{T}{j-1}
\wedge \trx{j-1,j}$.
\begin{proposition}
\label{prop:bnd_form}
 Let $H_j$ be a formula (depending only on variables of $j$-th cut)
 such that $\prob{\Abs{W}{j-1}}{I_0 \wedge \Abs{T}{j}} \equiv H_j
 \wedge \prob{\Abs{W}{j-1}}{I_0 \wedge \ttrx{j}}$.  Then $H_j$ is a
 boundary formula for the pair (\ks, \ksrr{j}).
\end{proposition}

Proofs of the propositions are given in the appendix.

%
%
\begin{proposition}
\label{prop:bnd_form_pqe}
Let $T_{j-1,j} = \trx{j-1,j} \wedge R_{j-1,j}$. Let
$H_j$ be a formula such that $\prob{\Abs{W}{j-1}}{I_0 \wedge
  \ttrx{j} \wedge R_{j-1,j}} \equiv H_j \wedge
\prob{\Abs{W}{j-1}}{I_0 \wedge \ttrx{j}}$. Then
$H_j$ is a boundary formula for the pair (\ks, \ksrr{j}).
\end{proposition}

One can view $R_{j-1,j}$ as a formula specifying the ``difference''
between $T_{j-1,j}$ and \trx{j-1,j}
Proposition~\ref{prop:bnd_form_pqe} suggests that $H_j$ can be
obtained by taking $R_{j-1,j}$ out of the scope of quantifiers i.e. by
PQE.

%
%
\subsection{Building boundary formulas inductively}
\label{subsec:bf_induct}
Proposition~\ref{prop:bnd_form} suggests that adding a boundary
formula $H_j$ makes up for the difference between $T_{j-1,j}$ and
\trx{j-1,j}. Suppose that one relaxes transition relation in every
time frame. Let \Ttrx{j} denote $\trx{0,1} \wedge \dots \wedge
\trx{j-1,j}$.  Let \Abs{H}{j} denote $H_0 \wedge \dots \wedge H_j$
where $H_0 = I$ and $H_1,\dots,H_j$ are boundary formulas.  Then
the following proposition is true.
%
%
\begin{proposition}
\label{prop:ind_bnd_form}
$\prob{\Abs{W}{j-1}}{ I_0 \wedge \Abs{T}{j}} \equiv$
$\exists{\Abs{W}{j-1}}[\Abs{H}{j} \wedge $ $\Ttrx{j}]$.
\end{proposition}

Boundary formulas $H_0,\dots,H_m$ can be built by induction using the
following procedure. Let $T_{j-1,j} = \trx{j-1,j} \wedge R_{j-1,j}$,
$j=1,\dots,m$. (We assume that $R_{j-1,j}$ is different in different
time frames.)  Formula $H_0 = I$ and formula $H_j$, $0 < j \leq m$ is
obtained by taking $R_{j-1,j}$ out of the scope of quantifiers in
formula \prob{\Abs{W}{j-1}}{\Abs{H}{j-1} \wedge \Ttrx{j} \wedge
  R_{j-1,j}}.  That is $\prob{\Abs{W}{j-1}}{\Abs{H}{j-1} \wedge
  \Ttrx{j} \wedge R_{j-1,j}} \equiv H_j \wedge
\prob{\Abs{W}{j-1}}{\Abs{H}{j-1} \wedge \Ttrx{j}}$.  The correctness
of this procedure follows from Proposition~\ref{prop:corr_bfs} of the
appendix.

Note that the greater $k$, the larger the formula in which $R_{k-1,k}$
is taken out of the scope of quantifiers. This topic is discussed
in~\cite{tech_rep_ec_lor}. There we argue the following.
In~\cite{hvc-14}, we introduced a PQE algorithm based on the machinery
of D-sequents~\cite{fmcad12,fmcad13}. The growth of formula size
mentioned above will cripple the performance of the algorithm
of~\cite{hvc-14} since the latter lacks a few crucial techniques
e.g. D-sequent re-using.  However, if a PQE solver employs D-sequent
re-using, this problem will either go away completely or at least will
be greatly mitigated.

\section{An Algorithm Of PC By LoR}
\label{sec:alg_descr}
In this section, we describe an algorithm of PC by LoR called
\alg. This algorithm is meant only for systems that have the
stuttering feature. In Subsection~\ref{subsec:stuttering}, we explain
the advantages of systems with stuttering and show how stuttering can
be introduced by a minor modification of the system at hand if the
latter does not have it. Subsections~\ref{subsec:co_conds},
~\ref{subsec:one_two_conds}, ~\ref{subsec:three_four_conds} describe
the properties of boundary formulas maintained by \Alg to guarantee
its convergence. A description of the pseudo-code of \Alg is given in
Subsections~\ref{subsec:alg_descr} and~\ref{subsec:rem_bs}. The
correctness of \Alg is proved in Subsection~\ref{subsec:correctness}.

%
%
\subsection{Stuttering}
\label{subsec:stuttering}
Suppose that one needs to check that a property $P$ of a sequential
circuit $M$ holds. Let $T$ be the transition relation specified by $M$
and \Ks be the transition system defined by $T$ and a formula $I$
specifying the initial states (see Subsection~\ref{subsec:problem}).
The \Alg algorithm described in this section is based on the
assumption that \Ks has the \tb{stuttering} feature i.e. \Ks can stay
in a given state arbitrarily long. This means that for every present
state \pnt{s}, there is an input assignment \pnt{x} such that the next
state produced by circuit $M$ is also \pnt{s}.  If \Ks does not have
this feature, one can introduce stuttering by adding to circuit $M$ a
combinational input variable $v$.  The modified circuit $M$ works as
before if $v=1$ and copies its current state to the output state
variables if $v=0$.  On the one hand, introduction of stuttering does
not affect the reachability of a bad state. On the other hand,
stuttering guarantees that \Ks has two nice properties.  First,
$\prob{W}{T(S,X,Y,S')} \equiv 1$ holds where $W = S \cup X \cup Y$.
Indeed for every next state \pnt{s'}, $T$ specifies a ``stuttering
transition'' from \pnt{s} to \pnt{s'} where \pnt{s} = \pnt{s'}.
Second, if a state is unreachable in \Ks after $n$ transitions it is
also unreachable after $m$ transitions if $m < n$.

%
%

\subsection{Four properties to guarantee convergence}
\label{subsec:co_conds}
The essence of \Alg is to build a boundary formula $H_j$ for every
time frame.  Boundary formulas are generated by \Alg one by one.  We
assume that $H_0$ is set to $I$. Let $\Abs{T}{j}$ denote $T_{0,1}
\wedge \dots \wedge T_{j-1,j}$, $j > 0$ and $\Abs{T}{0} \equiv 1$.  We
will refer to the four conditions below as \tb{CO conditions} (where CO
stands for Convergence of Over-approximations).
\begin{enumerate}
\item \impl{I}{H_j}
\item \impl{H_j}{P},
\item \impl{H_{j-1} \wedge \trx{j-1,j}}{H_j}, 
\item \impl{H_{j-1}}{H_j}, 
\end{enumerate}

(When we write formulas like \impl{I}{H_j} we assume that the sets of
variables are unified for the left and right parts of the
implication. That is \impl{I}{H_j} actually means
\impl{I(S)}{H_j(S)}.)  The CO conditions are similar to those imposed
on formulas $F_i$ specifying supersets of reachable states in
IC3~\cite{ic3}. However, formulas $H_j$ are built via relaxation of
transition relation and PQE i.e. quite differently from $F_i$ of IC3.
The convenience of the CO conditions is that no matter how formulas
satisfying these conditions are built, eventually a counterexample or
an inductive invariant are generated.
%
%
\subsection{Providing first and second CO conditions}
\label{subsec:one_two_conds}
The first CO condition of Subsection~\ref{subsec:co_conds} is achieved
as follows. Formula $H_j$ is built by resolving clauses of $I_0 \wedge
\Abs{T}{j}$. So $H_j$ is implied by $I_0 \wedge \Abs{T}{j}$. Due to
the stuttering feature, this means that $H_j$ is also implied by $I$
alone.

The second CO condition is provided in two steps. Suppose that all
boundary formulas up to $H_{j-1}$ already satisfy the CO conditions
and \Alg starts building formula $H_j$. In the first step, \Alg checks
if $H_{j-1} \wedge T_{j-1,j} \rightarrow P$. If not, then there is an
$H_{j-1}$-state \pnt{s_{j-1}} that reaches a bad state in one
transition. \Alg tries to strengthen $H_{j-1}$ by conjoining the
latter with a CNF formula $G$ falsified by \pnt{s_{j-1}}. To derive
this formula, \Alg calls procedure $\mi{RemBadSt}$ described in
Subsection~\ref{subsec:rem_bs}. It either generates a trace leading
to \pnt{s_{j-1}} (which means that $P$ fails) or returns formula $G$
above.  Formula $G$ is built by relaxing transition relations of some
previous time frames even more and strengthening boundary formulas of
those time frames to make up for such additional relaxation.

%
%
\setlength{\intextsep}{2pt}
\begin{figure}
\small
\begin{tabbing}
aaa\=bb\=cc\= dd\= \kill
$\mi{PC\_LoR}(T,I,P)$ \{\\
\tb{\scriptsize{1}}\> $H_0 := I$;  \\
\tb{\scriptsize{2}}\> $j = 1$; \\
\tb{\scriptsize{3}}\>  while ($\mi{true}$) \{ \\
\tb{\scriptsize{4}}\Tt  $\trx{j-1,j} := T_{j-1,j}$; \\
\tb{\scriptsize{5}}\Tt  $H_j := 1$; \\
\tb{\scriptsize{6}}\Tt  $\mi{Cex} := \mi{RemBadSt}(\Abs{H}{j},\Ttrx{j},P,j)$;\\
\tb{\scriptsize{7}}\Tt   if ($\mi{Cex} \neq \mi{nil}$) return($\mi{No}$); \\
\tb{\scriptsize{8}}\Tt  $\mi{FinRlx}(\Abs{H}{j},\Ttrx{j})$\\
\tb{\scriptsize{9}}\Tt  $\mi{ThirdCOcond}(\Abs{H}{j},\ttrx{j})$; \\
\tb{\scriptsize{10}}\Tt  $\mi{Inv} :=\mi{FinTouch}(\Abs{H}{j},\ttrx{j})$ ;\\
\tb{\scriptsize{11}}\Tt  if ($\mi{Inv}$) return($\mi{Yes}$);  \\
\tb{\scriptsize{12}}\Tt  $j := j+1;$ \}\}\\
\end{tabbing} 
\vspace{-15pt}
\caption{\Alg procedure}
\label{fig:pc_lor}
\end{figure}

The second step starts when $H_{j-1} \wedge T_{j-1,j} \rightarrow P$
holds. In this step, \Alg calls procedure $\mi{FinRlx}$ that relaxes
transition relation $T_{j-1,j}$ of $(j-1)$-th time frame and builds
formula $H_j$ implying $P$ that makes up for relaxing
$T_{j-1,j}$. Originally, $H_j = 1$ and $\trx{j-1,j} = T_{j-1,j}$. If
there is an $H_j$-state \pnt{s_j} that falsifies $P$, \Alg relaxes the
current transition relation \trx{j-1,j} to make \pnt{s_j} reachable
from an $H_{j-1}$-state. This relaxation has the form $\trx{j-1,j} =
\SPb{T}{*rlx}{j-1,j} \wedge R_{j-1,j}$ where \trx{j-1,j} is the
current relaxed formula and \SPb{T}{*rlx}{j-1,j} is a new one that
makes \pnt{s_j} reachable. \Alg looks for a formula $G$ such that
$\prob{\Abs{W}{j-1}}{I_0 \wedge \Abs{H}{j} \wedge \Ttrx{j-1,j} \wedge
  R_{j-1,j}} \equiv G \wedge \prob{\Abs{W}{j-1}}{I_0 \wedge \Abs{H}{j}
  \wedge \Ttrx{j-1,j}}$. Here \Ttrx{j-1,j} is equal to $\Ttrx{j-2,j-1}
\wedge \SPb{T}{*rlx}{j-1,j}$. Formula $G$ is falsified by \pnt{s_j}
and is conjoined with $H_j$ to exclude this state. This goes on until
$H_j$ implies $P$.

%
%
\subsection{Providing third and fourth CO conditions}
\label{subsec:three_four_conds}
After $H_j$ is generated as described above, it satisfies the first
two CO conditions of Subsection~\ref{subsec:co_conds} but the third
condition, in general, does not hold, i.e. \nimpl{H_{j-1} \wedge
  \trx{j-1,j}}{H_j}.  This happens if a clause of an \mbox{$m$-th}
time frame where $m < j-1$ is employed by procedure $\mi{FinRlx}$
above when generating $H_j$ that implies $P$. Let \pnt{s_{j-1}} be an
$H_{j-1}$-state that is one transition away from a state falsifying
$H_j$.  Then \Alg derives a formula falsified by \pnt{s_{j-1}} and
conjoins it with $H_{j-1}$. This formula is derived by the procedure
above used to eliminate $H_{j-1}$-states that are one transition away
from a bad state.  This goes on until \impl{H_{j-1} \wedge
  \trx{j-1,j}}{H_j} holds. Even if the third condition \impl{H_{m-1}
  \wedge \trx{m-1,m}}{H_m} holds for $m < j$, it may get broken after
adding clauses to formula $H_m$. Then the procedure above is used to
eliminate $H_{m-1}$-states that are one transition away from states
falsifying $H_m$.

The fourth CO condition is very easy to maintain. Due to the
stuttering feature, $\prob{\Abs{W}{m-1}}{I_0 \wedge \Abs{T}{m}} \rightarrow
\prob{\Abs{W}{j-1}}{I_0 \wedge \Abs{T}{j}}$, $m < j$.  So every clause of
$H_j$ can be added to every boundary formula $H_m$, $m < j$.

%
%
\subsection{Pseudo-code of \alg}
\label{subsec:alg_descr}
%
%
\setlength{\intextsep}{2pt}
\begin{figure}
\small
\begin{tabbing}
aaa\=bb\=cc\= dd\= \kill
$\mi{RemBadSt}(\Abs{H}{j},\Ttrx{j},P,j)$ \{\\
\tb{\scriptsize{1}}\> $\mi{Cex} := \emptyset$; \\
\tb{\scriptsize{2}}\> $\mi{length} := 0$; \\
\tb{\scriptsize{3}}\> while ($\mi{true}$) \{ \\
\tb{\scriptsize{4}}\Tt if ($\mi{length} = 0$) \{ \\
\tb{\scriptsize{5}}\ttt  $(\pnt{s_{j-1}},\pnt{s_j}) := \mi{FndBadSt}(H_{j-1} \wedge \trx{j-1,j}  \wedge \overline{P})$; \\
\tb{\scriptsize{6}}\ttt  if $((\pnt{s_{j-1}},\pnt{s_j}) = \mi{nil})$ return($\mi{nil}$);\\  
\tb{\scriptsize{7}}\ttt  $\mi{Cex} := (\pnt{s_{j-1}},\pnt{s_j})$;\\ 
\tb{\scriptsize{8}}\ttt   $\mi{length} := 2$;\\
\tb{\scriptsize{9}}\ttt  continue; \}\\
$~~~~ --------- $ \\
\tb{\scriptsize{10}}\Tt  $k := j-\mi{length}+1$;  \\
\tb{\scriptsize{11}}\Tt if ($k = 0$) return($\mi{Cex}$); \\
\tb{\scriptsize{12}}\Tt $\pnt{s_k} := \mi{FirstState}(\mi{Cex})$; \\
\tb{\scriptsize{13}}\Tt $\pnt{s_{k-1}} := \mi{ExtCex}(H_{k-1} \wedge \trx{k-1,k},\pnt{s_k})$; \\
$~~~~ --------- $ \\
\tb{\scriptsize{14}}\Tt if ($\pnt{s_{k-1}} \neq \mi{nil}$) \{\\
\tb{\scriptsize{15}}\ttt   $\mi{Cex} := (\pnt{s_{k-1}},\mi{Cex})$;\\
\tb{\scriptsize{16}}\ttt $\mi{length} := \mi{length}+1$; \\
\tb{\scriptsize{17}}\ttt continue; \}  \\
$~~~~ --------- $ \\
\tb{\scriptsize{18}}\Tt  $R_{k-1,k} := \mi{Relax}(H_{k-1},\trx{k-1,k},s_k)$;\\
\tb{\scriptsize{19}}\Tt  $\mi{PQE}(R_{k-1,k},H_k,\trx{k-1,k},\Abs{H}{k-1},\Ttrx{k-1})$;\\
\tb{\scriptsize{20}}\Tt  $\mi{RemFrstSt}(\mi{Cex},\pnt{s_k})$; \\
\tb{\scriptsize{21}}\Tt  $\mi{length} := \mi{length}-1$;\}\}\\
\end{tabbing} 
\vspace{-15pt}
\caption{\ti{RemBadSt} procedure}
\label{fig:rem_bad_st}
\end{figure}

The pseudo-code of \Alg is given
Figure~\ref{fig:pc_lor}. Boundary formulas are derived in the while
loop (lines 3-12). In every iteration, a boundary formula $H_j$ is
derived and $j$ is incremented by one. Originally, $H_0$ is set to $I$
and $j$ is set to 1.

%
%
\setlength{\intextsep}{2pt}
\begin{figure}
\small
\begin{tabbing}
aaa\=bb\=cc\= dd\= \kill
$\mi{FinRlx}(\Abs{H}{j},\Ttrx{j})$ \{\\
\tb{\scriptsize{1}}\> while ($\mi{true}$) \{\\
\tb{\scriptsize{2}}\Tt   $\pnt{s} := \mi{FindSat}(H_j \wedge \overline{P})$;\\
\tb{\scriptsize{3}}\Tt  if ($\pnt{s} = \mi{nil}$) return;\\
\tb{\scriptsize{4}}\Tt  $R_{k-1,k} := \mi{Relax}(H_{j-1},\trx{j-1,j},s)$;\\
\tb{\scriptsize{5}}\Tt  $\mi{PQE}(R_{j-1,j},H_j,\trx{j-1,j},\Abs{H}{j-1},\Ttrx{j-1})$;\}\}\\
\end{tabbing} 
\vspace{-15pt}
\caption{\ti{FinRlx} procedure}
\label{fig:fin_rlx}
\end{figure}

Every iteration starts by making $H_j$ satisfy the second CO condition
i.e. \impl{H_j}{P} (lines 4-8). First \Alg calls procedure
$\mi{RemBadSt}$ that either returns a counterexample or strengthens
formula $H_{j-1}$ to guarantee $H_{j-1} \wedge T_{j-1,j} \rightarrow
P$.  Procedure $\mi{RemBadSt}$ is described in detail in
Subsection~\ref{subsec:rem_bs}. If $\mi{RemBadSt}$ returns a
counterexample $\mi{Cex}$, \Alg terminates reporting that property $P$
failed. Otherwise, \Alg calls procedure $\mi{FinRlx}$ shown in
Figure~\ref{fig:fin_rlx}. Its work was described in
Subsection~\ref{subsec:one_two_conds}: $\mi{FinRlx}$ relaxes the
transition relation of $(j-1)$-th time frame and adds clauses making
up for this relaxation to $H_j$ until \impl{H_j}{P} holds.

To make sure that \impl{H_{j-1} \wedge \trx{j}}{H_j} holds, \Alg calls
procedure $\mi{ThirdCOcond}$ that works as described in
Subsection~\ref{subsec:three_four_conds}. Finally, to guarantee that
\impl{H_{m-1}}{H_m} holds for all $1 \leq m \leq j$, procedure
$\mi{FinTouch}$ is called.  First, this procedure tries to push every
clause $C$ of $H_j$ to previous boundary formulas. If $C$ is not
implied by $H_m$, $m < j$, it is added to $H_m$ and $\mi{FinTouch}$
tries to push $C$ to $H_{m-1}$. Otherwise, the process of pushing
clause $C$ stops: if $C$ is implied by $H_m$ it is also implied by
every formula $H_k$, $k < m$. The process of pushing clauses of $H_j$
may break third CO condition for some boundary formulas. In this case,
the $\mi{ThirdCOcond}$ procedure is called to repair this
condition. Eventually, $\mi{FinTouch}$ makes all boundary formulas
$H_i$, $i=0,\dots,j$ meet third and fourth CO conditions.

Procedure $\mi{FinTouch}$ also checks if \impl{H_m}{H_{m-1}} holds for
some $m$, $0 \leq m \leq j$. If so, then $H_{m-1} \equiv H_m$ and
$H_{m-1}$ is an inductive invariant (see the proof of
Proposition~\ref{prop:sound}). Checking for presence of an inductive
invariant by testing logical implication is harder than by checking
syntactic equivalence performed in IC3. However, one can use
optimization to mitigate this problem. Here is an example of such
optimization. Formula $H_m$ implies $H_{m-1}$ iff every clause of
$H_{m-1}$ is implied by $H_m$. If a clause of $H_{m-1}$ is implied by
$H_m$, it remains implied no matter what clauses are added to
$H_{m-1}$ and $H_m$. So when checking if \impl{H_m}{H_{m-1}} holds, it
suffices to check for implication every clause $C \in H_{m-1}$ that is
not marked as implied by $H_m$ yet. If $C$ is implied by $H_m$, it is
marked to avoid testing it in the future. Otherwise,
\impl{H_{m}}{H_{m-1}} does not hold and no testing of other unmarked
clauses of $H_{m-1}$ is necessary.

%
%
\subsection{Description of $\mi{RemBadSt}$ procedure}
\label{subsec:rem_bs}
The pseudo-code of the $\mi{RemBadSt}$ procedure is given in
Figure~\ref{fig:rem_bad_st}. The goal of $\mi{RemBadSt}$ is to
strengthen boundary formula $H_{j-1}$ so that $H_{j-1} \wedge
T_{j,j-1} \rightarrow P$ holds. This is the first step of generation
of formula $H_j$ that implies $P$ (see
Subsection~\ref{subsec:one_two_conds}). If $H_{j-1}$ cannot be
strengthened to guarantee the condition above, then a counterexample of
length $j$ is generated by $\mi{RemBadSt}$.

All the work is done in a while loop (lines 3-21) where
$\mi{RemBadSt}$ tries to construct a counterexample. This
counterexample is built in reverse from a bad state reachable from an
$H_{j-1}$-state. In every iteration of the loop, $\mi{RemBadSt}$
either extends the current trace by one more state or shows that the
last $H_m$-state of the trace cannot be reached from an
$H_{m-1}$-state. The latter triggers tightening up formula $H_m$ after
additional relaxation of the current transition \trx{m-1,m}.  The
length of the current trace is specified by variable $\mi{length}$.
The body of the while loop can be partitioned into parts separated by
the dotted lines in Figure~\ref{fig:rem_bad_st}.  If $\mi{length} =
0$, the current trace is empty and $\mi{RemBadSt}$ tries to initialize
it (lines 5-9). Namely, it looks for an $H_{j-1}$-state \pnt{s_{j-1}}
that is one transition away from a bad state \pnt{s_j}. If such states
\pnt{s_{j-1}} and \pnt{s_j} are found, counterexample $\mi{Cex}$ is
initialized with (\pnt{s_{j-1}},\pnt{s_j}).  Otherwise,
$\mi{RemBadSt}$ returns $\mi{nil}$ reporting that $H_{j-1} \wedge
T_{j,j-1} \rightarrow P$ holds.

If $\mi{length} > 0$, $\mi{RemBadSt}$ tries to extend the current
trace (lines 10-13). Let \pnt{s_k} be the state added to $\mi{Cex}$
the last.  If $k = 0$ i.e. if \pnt{s_k} is an $I$-state, the current
$\mi{Cex}$ is a counterexample and $\mi{RemBadSt}$ terminates
returning $\mi{Cex}$.  Otherwise, $\mi{RemBadSt}$ tries to find an
$H_{k-1}$-state \pnt{s_{k-1}} that is one transition away from
\pnt{s_k}. If $\mi{RemBadSt}$ succeeds, $\mi{Cex}$ is extended by
\pnt{s_{k-1}} (lines 15-16). 

If $\mi{RemBadSt}$ fails to find \pnt{s_{k-1}}, $\mi{Cex}$ cannot be
extended to a counterexample. Then $\mi{RemBadSt}$ does the following
(lines 18-21). The current transition relation \trx{k-1,k} is relaxed
even more as described in
Subsection~\ref{subsec:one_two_conds}. Namely, \trx{k-1,k} is
represented as $\SPb{T}{*rlx}{k-1,k} \wedge R_{k-1,k}$ where
\SPb{T}{*rlx}{k-1,k} is a new transition relation for $(k\!-\!\!1)$-th
time frame that makes \pnt{s_k} reachable. $\mi{RemBadSt}$ calls a
PQE-solver to build a formula $G$ such that $\prob{\Abs{W}{k-1}}{I_0
  \wedge \Abs{H}{k} \wedge \Ttrx{k-1,k} \wedge R_{k-1,k}} \equiv G
\wedge \prob{\Abs{W}{k-1}}{I_0 \wedge \Abs{H}{k} \wedge
  \Ttrx{k-1,k}}$. Here \Ttrx{k-1,k} is equal to $\Ttrx{k-2,k-1} \wedge
\SPb{T}{*rlx}{k-1,k}$. Formula $G$ is falsified by \pnt{s_k} and so is
conjoined with $H_k$ to exclude this state. Then $\mi{RemBadSt}$
removes \pnt{s_k} from $\mi{Cex}$ and starts a new iteration.

\subsection{Correctness of \alg}
\label{subsec:correctness}
This subsection lists propositions proving correctness of
\alg.
%
%
\setcounter{proposition}{4}
\begin{proposition}
\label{prop:lmtd_corr}
Let $H_j$, $j=1,\dots,m$ be formulas derived by \Alg for $m$ time
frames where \impl{H_j}{P}.  Then property $P$ holds for system \Ks
for at least $m$ transitions.
\end{proposition}
%
%
\begin{proposition}
\label{prop:sound}
\Alg is sound.
\end{proposition}
%
%
\begin{proposition}
\label{prop:complete}
\Alg is complete.
\end{proposition}

\section{Two Important Modifications Of \alg}
\label{sec:two_mods}
In this section, we consider two modifications of the \Alg algorithm
described in Section~\ref{sec:alg_descr}. The first modification is to
incorporate a ``manual'' relaxation that exploits the semantics of the
system.  The second modification is to combine LoR with the machinery
of inductive clauses of IC3. Sequential equivalence checking is a
promising application of the second modification.

\subsection{Relaxation by an educated guess}
\label{subsec:str_awe_rlx}
In this subsection, we describe a modification of \Alg that starts
building a boundary formula $H_j$ by a relaxation that is just a guess
tailored to a particular class of systems/properties.  An example of
such a relaxation is given in Section~\ref{sec:exmp}.  The pseudo-code
of modified \Alg is given in Figure~\ref{fig:str_ind_rlx}. The only
difference between the original version shown in Fig.~\ref{fig:pc_lor}
and the modified one is in line 5 where function $\mi{EducatGuessRlx}$
is called instead of setting $H_j$ to 1. This function does the
following.  First it represents the original transition relation
$T_{j-1,j}$ as $\trx{j-1,j} \wedge R_{j-1,j}$. Here \trx{j-1,j} is the
relaxed transition relation replacing $T_{j-1,j}$.  Then
$\mi{EducatGuessRlx}$ calls a PQE solver to build a formula $H_j$ such
that $\prob{\Abs{W}{j-1}}{I_0 \wedge \Abs{H}{j-1} \wedge \Ttrx{j-1,j}
  \wedge R_{j-1,j}} \equiv H_j \wedge \prob{\Abs{W}{j-1}}{I_0 \wedge
  \Abs{H}{j-1} \wedge \Ttrx{j-1,j}}$. Here \Ttrx{j-1,j} is equal to
$\Ttrx{j-2,j-1} \wedge \trx{j-1}$.

%
%
\setlength{\intextsep}{2pt}
\begin{figure}
\small
\begin{tabbing}
aaa\=bb\=cc\= dd\= \kill
$\mi{PC\_LoR}(T,I,P)$ \{\\

$~~~~..... $ \\
\tb{\scriptsize{4}}\Tt  $\trx{j-1,j} := T_{j-1,j}$; \\
\tb{\scriptsize{5*}} \Tt  $H_j := \mi{EducatGuessRlx}(\Abs{H}{j-1},\Ttrx{j})$; \\  
\tb{\scriptsize{6}}\Tt  $\mi{Cex} := \mi{RemBadSt}(\Abs{H}{j},\Ttrx{j},P,j)$;\\
$~~~~..... $ \\

\end{tabbing} 
\vspace{-15pt}
\caption{\ti{PC\_LoR} plus relaxation by an educated guess}
\label{fig:str_ind_rlx}
\end{figure}

\subsection{Combining LoR with machinery of inductive clauses}
\label{subsec:lor_ic3}
In this subsection, we describe an algorithm called \Mix (IC stands
for Inductive Clauses) that combines LoR and the machinery of
inductive clauses introduced by IC3~\cite{ic3}.  Given a transition
relation $T$, clause $C$ is called \tb{inductive} with respect to
formula $F$ if $F(S) \wedge C(S) \wedge T(S,X,Y,S') \rightarrow C(S')$
holds.  Our interest in \Mix is twofold. First, when computing a new
boundary formula, \Alg often has to go far back to tighten boundary
formulas computed earlier. This tightening is done to make up for
additional relaxation of transition relations of previous time
frames. The great performance of IC3 suggests that tightening of
boundary formulas of previous time frames can be efficiently done by
adding inductive clauses. Second, IC3 builds an inductive invariant by
tightening property $P$ with inductive clauses. This may result in
poor performance if an inductive invariant is ``far away'' from $P$.
Sequential equivalence checking is an example of a PC problem where
IC3 may perform poorly (see Section~\ref{sec:exmp}). \Mix is meant to
address this issue.

%
%
\setlength{\intextsep}{2pt}
\begin{figure}
\small
\begin{tabbing}
aaa\=bb\=cc\= dd\= \kill
$\mi{LoR\_IC}(T,I,P)$ \{\\
\tb{\scriptsize{1}}\> $H_0 := I$;  \\
\tb{\scriptsize{2}}\> $j = 1$; \\
\tb{\scriptsize{3}}\>  while ($\mi{true}$) \{ \\
\tb{\scriptsize{4}}\Tt  $\trx{j-1,j} := T_{j-1,j}$; \\
\tb{\scriptsize{5}} \Tt $H_j := 1$; \\
\tb{\scriptsize{6*}}\Tt  $\mi{Cex} := \SUb{RemBadSt}{IC}(\Abs{H}{j},\Ttrx{j},P,j)$;\\

\tb{\scriptsize{7}}\Tt   if ($\mi{Cex} \neq \mi{nil}$) return($\mi{No}$); \\
\tb{\scriptsize{8}}\Tt  $\mi{FinRlx}(\Abs{H}{j},\Ttrx{j})$\\
\tb{\scriptsize{9*}}\Tt  $\SUb{ThirdCOcond}{IC}(\Abs{H}{j},\ttrx{j})$; \\
\tb{\scriptsize{10*}}\Tt  $\mi{Inv} :=\SUb{FinTouch}{IC}(\Abs{H}{j},\ttrx{j})$ ;\\
\tb{\scriptsize{11}}\Tt  if ($\mi{Inv}$) return($\mi{Yes}$);  \\
\tb{\scriptsize{12}}\Tt  $j := j+1;$ \}\}\\
\end{tabbing} 
\vspace{-15pt}
\caption{\ti{LoR\_IC} procedure}
\label{fig:lor_ic}
\end{figure}

The pseudo-code of \Mix is shown in Figure~\ref{fig:lor_ic}. The
lines where \Mix is different from \Alg are marked with an asterisk.
Consider how \Mix builds formula $H_j$ after formulas
$H_0,\dots,H_{j-1}$ satisfying the four CO conditions have been
generated. Similarly to \alg, \Mix makes two steps to guarantee that
the second CO condition i.e. \impl{H_j}{P} holds. In the first step,
it makes sure that $H_{j-1} \wedge \trx{j-1,j} \rightarrow P$
holds. However, in contrast to \alg, this is done by calling procedure
$\SUb{RemBadSt}{IC}$ generating inductive clauses. If there is an
$H_{j-1}$-state \pnt{s_{j-1}} that is one transition away from a bad
state, a clause $C$ inductive with respect to $H_{j-1}$ is
generated. This clause is falsified by \pnt{s_{j-1}} and so is added
to $H_{j-1}$ to exclude this state. The second step is performed like
in \Alg by calling procedure $\mi{FinRlx}$. The latter relaxes
transition relation $T_{j-1,j}$ and builds $H_j$ as a set set of
clauses making up for this relaxation. This is where \Mix is
\ti{different} from IC3.  In IC3, formula $H_{j}$ is built by
conjoining the inductive clauses generated to exclude $H_{j-1}$-states
with $P$. Note that using these clauses when forming $H_j$ is not
actually mandatory. The ``why-not'' argument given in~\cite{ic3} is
that these clauses are implied by $H_{j-1} \wedge T_{j-1,j}$.

To satisfy the third CO condition, \Mix calls function
\SUb{ThirdCOcond}{IC}.  In contrast to $\mi{ThirdCOcond}$ of \alg,
\SUb{ThirdCOcond}{IC} does the job by generation of inductive
clauses. Suppose that one needs to eliminate an $H_{j-1}$-state
\pnt{s_{j-1}} from which a state \pnt{s_j} falsifying $H_{j}$ is
reachable in one transition.  Then $\SUb{ThirdCOcond}{IC}$ generates a
clause inductive with respect to $H_{j-1}$. This clause is falsified
by \pnt{s_{j-1}} and so is added to $H_{j-1}$ to exclude
\pnt{s_{j-1}}. To guarantee that the third and fourth CO conditions
hold for all formulas $H_k$, $k=0,\dots,j$ built so far, \Mix calls
function \SUb{FinTouch}{IC}. In contrast to $\mi{FinTouch}$,
\SUb{FinTouch}{IC} does the job via inductive clauses.

Note that instead of initializing $H_j$ to 1 (line 5 of
Fig.~\ref{fig:lor_ic}), one can call procedure
$\mi{EducatGuessRlx}$ to apply a transition relation relaxation
tailored to a particular system/property (see
Subsection~\ref{subsec:str_awe_rlx}).  We believe that the version of
\Mix where $\mi{EducatGuessRlx}$ employs relaxation described in
Section~\ref{sec:exmp} is a promising algorithm for sequential
equivalence checking. The idea here is as follows. First,
$\mi{EducatGuessRlx}$ generates formula $H_j$ that is close to an
inductive invariant. Then some fine-tuning of $H_j$ is done by adding
inductive clauses generated when computing boundary formulas $H_m$, $m
> j$.

\section{Conclusions}
\label{sec:conclusions}
We introduced a new framework for Property Checking (PC) based on a
method called Logic Relaxation (LoR). The appeal of PC by LoR is that
an inductive invariant is the result of comparison of the original and
relaxed transition systems. So the complexity of PC can be
significantly reduced if the relaxed system is close to the original
one.  A key part of the LoR method is a technique called partial
quantifier elimination.  So it is extremely important to keep
improving the performance of algorithms implementing this technique.

\vspace{15pt}
\appendix
\setcounter{proposition}{0}
%
%
\begin{lemma}
\label{lemma:ec}
Let $\mi{EQ}^X_0$ denote $\mi{EQ}(X^N_0,X^K_0)$. Let $T_{0,1} =
T^N_{0,1} \wedge T^K_{0,1} \wedge \mi{EQ}^X_0$ be the transition
relation specifying miter $M$ of two identical circuits $N$ and $K$ in
terms of initial time frame variables. Let the relaxed
transition \trx{0,1} for miter $M$ be equal to $T^N_{0,1} \wedge
T^K_{0,1}$ i.e. $T_{0,1} = \mi{EQ}^X_0 \wedge \trx{0,1}$. Let
$\mi{EQ}^S_1$ denote $\mi{EQ}(S^N_1,S^K_1)$. Let $s^N_{j,1}$ denote
$j$-th state variable of circuit $N$ of time frame 1. Then
\begin{enumerate}[a)]
\item $I_0 \wedge T_{0,1} \rightarrow \mi{EQ}^S_1$
\item $I_0 \wedge \trx{0,1} \rightarrow (s^N_{j,1} \equiv s^K_{j,1})$ if
\vspace{2pt}
the value of variable $s^N_{j,1}$ remains the same for every state
reachable from an $I^N_0$-state in one transition.
\item  $\prob{W_0}{I_0 \wedge
\mi{EQ}^X_0 \wedge \trx{0,1}} \equiv \mi{EQ}^S_1 \wedge \prob{W_0}{I_0 \wedge \trx{0,1}}$
\end{enumerate}
\end{lemma}
\begin{proof}[\kern-10pt Proof]

\ti{Item a)}. Since $N$ and $K$ are identical and $I_0$ implies $\mi{EQ}^S_1$, 
miter $M$ reaches only states (\pnt{s_N},\pnt{s_M}) where
\pnt{s_N} = \pnt{s_M}. This means that $I_0 \wedge T_{0,1} \rightarrow \mi{EQ}^S_1$.

\ti{Item b)}. This item explains under what conditions some clauses of $\mi{EQ}^X_0$
are redundant without adding any clauses of $\mi{EQ}^S_1$.  Suppose
that the assumption of item b) holds.  Then the fact that \trx{0,1}
does not impose restrictions on $X^N$ and $X^K$ does not matter as far
as variables $s^N_{j,1}$ and $s^K_{j,1}$ are concerned. Indeed, the
value of those variables remains the same for all assignments to $X_N$
and $X_K$.  This means that $I_0 \wedge \trx{0,1} \rightarrow
(s^N_{j,1} \equiv s^K_{j,1})$.  So when taking formula $\mi{EQ}^X_0$
out of the scope of quantifiers, adding all the clauses of
$\mi{EQ}^S_1$ is not necessary. (Recall that given sets of Boolean
variables $A=(a_1,\dots,a_k)$ and $B=(b_1,\dots,b_k)$, $\mi{EQ}(A,B) =
(a_1 \equiv b_1) \wedge \dots \wedge (a_k \equiv b_k)$). Namely, one
does not need to add the clauses of $\mi{EQ}^S_1$ specifying
$s^N_{j,1} \equiv s^K_{j,1}$.

\ti{Item c)}. Assume the contrary. Taking into account that $T_{0,1} = \mi{EQ}^X_0 \wedge \trx{0,1}$,
this means that $\prob{W_0}{I_0 \wedge
T_{0,1}} \not\equiv \mi{EQ}^S_1 \wedge \prob{W_0}{I_0 \wedge \trx{0,1}}$.
Let $\mi{Left\_part}$ and $\mi{Right\_part}$ specify the left and
right parts of the inequality above respectively. Consider the two
alternatives.

\vspace{3pt}
\noindent $\mi{Left\_part}$ = 1, $\mi{Right\_part}$ = 0. 
Then there is an assignment \pnt{t} to $W_0 \cup S_1$
that satisfies $I_0 \wedge T_{0,1}$ and hence
$I_0 \wedge \trx{0,1}$. Since $\mi{Right\_part}=0$, then
$\mi{EQ}^S_1(\pnt{t}) = 0$, which means that $I_0 \wedge
T_{0,1} \not\rightarrow \mi{EQ}^S_1$. So we have a contradiction.

\vspace{3pt}
\noindent $\mi{Left\_part}$ = 0, $\mi{Right\_part}$ = 1 for an
assignment \pnt{s_1} to $S_1$.  Then there is an assignment \pnt{t} to
$W_0 \cup S_1$ obtained by extending
\pnt{s_1} that satisfies
$\mi{EQ}^S_1 \wedge I_0 \wedge \trx{0,1}$. Since $\mi{Left\_part} =
0$, then \pnt{t} falsifies $I_0 \wedge T_{0,1}$. This means
that \pnt{t} falsifies $\mi{EQ}^X_0$
i.e. $\pnt{x^N_0} \neq \pnt{x^K_0}$ where \pnt{x^N_0} and \pnt{x^K_0}
are assignments to $X^N_0$ and $X^K_0$ from \pnt{t}. Let \pnt{t^*} be
the assignment obtained from \pnt{t} by replacing assignment to
$X^K_0 \cup Y^K_0$ specifying the execution trace for \pnt{x^K_0} with
that specifying the execution trace for input \pnt{x^N_0}. It is not
hard to see that \pnt{t^*} has the same assignment to $S_1$ as \pnt{t}
but satisfies $I_0 \wedge T_{0,1}$. So $\mi{Left\_part} = 1$ for
assignment \pnt{s_1} to $S_1$ and we have a contradiction.
\end{proof}

\vspace{3pt}
%
%
\begin{proposition}
Let $H_j$ be a formula (depending only on variables of $j$-th cut)
such that $\prob{\Abs{W}{j-1}}{I_0 \wedge \Abs{T}{j}} \equiv
H_j \wedge \prob{\Abs{W}{j-1}}{I_0 \wedge \ttrx{j}}$.  Then $H_j$ is a
boundary formula for the pair (\ks, \ksrr{j}).
\end{proposition}

\begin{proof}[\kern-10pt Proof]
Assume the contrary i.e. $H_j$ is not a boundary formula.
Definition~\ref{def:bnd_form} suggests that then one of the two cases
below takes place.

\ti{Case 1:} There is a valid trace $t$=\tr{s}{0}{j} of \ksrr{j} such that
\pnt{s_j} is not reachable in \Ks in $j$ transitions and $H_j(\pnt{s_j}) = 1$.
Since $t$ is a valid trace in \Ksr and $H_j(\pnt{s_j})=1$, formula
$H_j \wedge \prob{\Abs{W}{j-1}}{I_0 \wedge \ttrx{j}}$ evaluates to 1
under assignment \pnt{s_j} to $S_j$.
Then \prob{\Abs{W}{j-1}}{I_0 \wedge \Abs{T}{j}} evaluates to 1
under \pnt{s_j} as well, which means that \pnt{s_j} is reachable
in \ks. So we have a contradiction.

\ti{Case 2:} There is a valid trace $t$=\tr{s}{0}{j} of \Ks and yet
$H_j(\pnt{s_j}) = 0$. Then
formula \prob{\Abs{W}{j-1}}{I_0 \wedge \Abs{T}{j}} evaluates to 1
under assignment \pnt{s_j}. On the other hand, the fact that
$H_j(\pnt{s_j})=0$ means that
\prob{\Abs{W}{j-1}}{I_0  \wedge \Abs{T}{j}}
 $\neq$  $H_j \wedge \prob{\Abs{W}{j-1}}{I_0 \wedge \ttrx{j}}$
 under assignment \pnt{s_j}.
So we have a contradiction.
\end{proof}

\vspace{3pt}
%
%
\begin{proposition}
Let $T_{j-1,j} = \trx{j-1,j} \wedge R_{j-1,j}$. Let $H_j$ be a formula
such that $\prob{\Abs{W}{j-1}}{I_0 \wedge \ttrx{j} \wedge
R_{j-1,j}} \equiv H_j \wedge
\prob{\Abs{W}{j-1}}{I_0 \wedge \ttrx{j}}$. Then
$H_j$ is a boundary formula for the pair (\ks, \ksrr{j}).
\end{proposition}
\begin{proof}[\kern-10pt Proof]
By definition, $\Abs{T}{j} = \ttrx{j} \wedge R_{j-1,j}$. Then the
correctness of the proposition follows from
Proposition~\ref{prop:bnd_form}.
\end{proof}

\vspace{3pt}
%
%
\begin{proposition}
$\prob{\Abs{W}{j-1}}{ I_0 \wedge \Abs{T}{j}} \equiv$
$\exists{\Abs{W}{j-1}}[\Abs{H}{j} \wedge $ $\Ttrx{j}]$.
\end{proposition}
\begin{proof}[\kern-10pt Proof]
Let us prove the proposition by induction.
Proposition~\ref{prop:bnd_form} entails that the proposition at hand
holds for $j=1$. Let us show that the correctness of the proposition
for $j > 1$, implies that it holds for $j+1$.  Let \Fii
denote \prob{\Abs{W}{j}}{ I_0 \wedge \Abs{T}{j+1}}. Formula \Fii can
be rewritten as \Prob{W_j}{\Abs{W}{j-1}}{I_0 \wedge \Abs{T}{j} \wedge
T_{j,j+1}}.  Taking into account that $T_{j,j+1}$ does not depend on
variables of \Abs{W}{j-1}, formula \Fii can represented as
$\exists{W_j}[T_{j,j+1} \wedge \prob{\Abs{W}{j-1}}{I_0 \wedge \Abs{T}{j}}]$.
Using the inductive hypothesis this formula can be transformed into
$\exists{W_j}[T_{j,j+1} \wedge \prob{\Abs{W}{j-1}}
{\Abs{H}{j} \wedge \Ttrx{j}}]$. Taking into account that $T_{j,j+1}
= \trx{j,j+1} \wedge R_{j,j+1}$, formula \Fii can be represented
as \prob{\Abs{W}{j}}{\Abs{H}{j} \wedge \Ttrx{j} \wedge \trx{j,j+1} \wedge
R_{j,j+1}}. Since $H_{j+1}$ is obtained by taking $R_{j,j+1}$ out of
the scope of quantifiers, formula \Fii can be rewritten as
$H_{j+1} \wedge \prob{\Abs{W}{j}}{\Abs{H}{j} \wedge \Ttrx{j} \wedge \trx{j,j+1}}$.
So the original formula \Fii is logically equivalent to
formula \prob{\Abs{W}{j}}{\Abs{H}{j+1} \wedge \Ttrx{j+1}}.
\end{proof}

\vspace{3pt}
%
%
\begin{proposition}
\label{prop:corr_bfs}
Let $T_{j-1,j} = \trx{j-1,j} \wedge R_{j-1,j}$, $j > 0$. Let formulas
$H_0,\dots,H_j$ be built consecutively as follows. Formula $H_0$
equals $I$ and formula $H_j$, $j > 0$ is built to satisfy
$\prob{\Abs{W}{j-1}}{\Abs{H}{j-1} \wedge \Ttrx{j} \wedge
R_{j-1,j}} \equiv H_j \wedge
\prob{\Abs{W}{j-1}}{\Abs{H}{j-1} \wedge \Ttrx{j}}$.
Then $H_0,\dots,H_j$ are boundary formulas.
\end{proposition}
\begin{proof}[\kern-10pt Proof]
The fact that $H_0$ is a boundary formula follows from
Definition~\ref{def:bnd_form}.  Let us show that formulas
$H_1,\dots,H_j$ are also boundary by induction. Assume that formulas
$H_1,\dots,H_{j-1}$ are boundary and show that then $H_j$ is a
boundary formula as well.

Proposition~\ref{prop:ind_bnd_form} entails that
formula \prob{\Abs{W}{j-2}}{\Abs{H}{j-1} \wedge \Ttrx{j-1}} can be
replaced with \prob{\Abs{W}{j-2}}{I_0 \wedge \Abs{T}{j-1}}. So
\prob{\Abs{W}{j-1}}{\Abs{H}{j-1} \wedge \Ttrx{j} \wedge R_{j-1,j}}
can be rewritten as
$\exists{W_{j-1}}[\prob{\Abs{W}{j-2}}{\Abs{H}{j-1} \wedge \Ttrx{j-1} \wedge \trx{j-1,j} \wedge
R_{j-1,j}}]$, then as
$\exists{W_{j-1}}[\prob{\Abs{W}{j-2}}{I_0 \wedge \Abs{T}{j-1} \wedge \trx{j-1,j} \wedge
R_{j-1,j}}]$ and finally
as \prob{\Abs{W}{j-1}}{I_0 \wedge \Abs{T}{j}}.  Similarly formula
$H_j \wedge
\prob{\Abs{W}{j-1}}{\Abs{H}{j-1} \wedge \Ttrx{j}}$ can be rewritten
as
$H_j \wedge \exists{W_{j-1}}[\prob{\Abs{W}{j-2}}{\Abs{H}{j-1} \wedge \Ttrx{j-1} \wedge \trx{j-1,j}}]$,
then as
$H_j \wedge \exists{W_{j-1}}[\prob{\Abs{W}{j-2}}{I_0 \wedge \Abs{T}{j-1} \wedge \trx{j-1,j}}]$
and finally as $H_j \wedge \prob{\Abs{W}{j-1}}{I_0 \wedge \ttrx{j}}$.
So $H_j$ satisfies $\prob{\Abs{W}{j-1}}{I_0 \wedge \Abs{T}{j}} \equiv
H_j \wedge \prob{\Abs{W}{j-1}}{I_0 \wedge \ttrx{j}}$. Then from
Proposition~\ref{prop:bnd_form} it follows that $H_j$ is a boundary
formula.
\end{proof}

\vspace{3pt}
%
%
\begin{proposition}
Let $H_j$, $j=1,\dots,m$ be formulas derived by \Alg for $m$ time
frames where \impl{H_j}{P}.  Then property $P$ holds for system \Ks
for at least $m$ transitions.
\end{proposition}
\begin{proof}[\kern-10pt Proof]
As we mentioned in Subsection~\ref{subsec:one_two_conds},
$I_0 \wedge \Abs{T}{j} \rightarrow H_j$ holds. Then \impl{H_j}{P}
entails $I_0 \wedge \Abs{T}{j} \rightarrow P$.
\end{proof}

\vspace{3pt}
%
%
\begin{proposition}
\Alg is sound.
\end{proposition}

\begin{proof}[\kern-10pt Proof]
Consider the two obvious alternatives.

\vspace{2pt}
\ti{The answer is ``property fails''}. This answer is returned by \Alg if 
there is an assignment \pnt{t} satisfying
$I_0 \wedge \Abs{H}{j} \wedge \Ttrx{j}$ and $\overline{P}$. From
Proposition~\ref{prop:ind_bnd_form} it follows, that then there is an
assignment \pnt{t^*} satisfying $I_0 \wedge \Abs{T}{j}$ and
$\overline{P}$. Hence there is a counterexample of length $j+1$.

\vspace{3pt}
\ti{The answer is ``property holds''}. 
This answer is returned when there appear a formula $H_{j-1}$ such
that \impl{H_j}{H_{j-1}} and $H_{j-1} \wedge \trx{j-1,j} \rightarrow
H_j$ hold.  Since $H_{j-1}$ implies $H_j$, then $H_{j-1} \equiv H_j$.
Since $T_{j-1,j}$ implies \trx{j-1,j}, $H_{j-1} \wedge
T_{j-1,j} \rightarrow H_j$ holds as well and $H_{j-1}$ is an inductive
invariant.
\end{proof}

\vspace{3pt}
%
%
\begin{proposition}
\Alg is complete.
\end{proposition}
\begin{proof}[\kern-10pt Proof] Consider the following alternatives.

\ti{Property $P$ fails}. Let $j$ be the first time
frame where a bad state \pnt{s_j} is reachable by \ks.  Let $H_j$ be a
boundary formula generated for $j$-th time frame.  Since $H_j$ is
implied by $I_0 \wedge \Abs{T}{j}$, \Alg will not be able to make
$H_j$ imply $P$. Then procedure $\mi{RemBadSt}$ will terminate
reporting that $P$ failed.

\vspace{5pt}
\ti{Property $P$ holds}. Let $H_0,\dots,H_m$, be
a sequence of boundary formulas built by \Alg. Let \impl{H_{j-1}}{H_j}
hold for every $j$, $0 < j \leq m$. If $m > 2^{|S|}$ where $S$ is the
set of state variables, there has to be a formula $H_{j-1}$ that is
logically equivalent to $H_j$. Since
$H_{j-1} \wedge \trx{j-1,j} \rightarrow H_j$ holds, \Alg will
terminate reporting that $P$ holds.
\end{proof}

\bibliographystyle{plain}
\bibliography{short_sat,local}

\begin{thebibliography}{10}

\bibitem{ic3}
A.~R. Bradley.
\newblock Sat-based model checking without unrolling.
\newblock In {\em VMCAI}, pages 70--87, 2011.

\bibitem{sat09}
E.~Goldberg.
\newblock Boundary points and resolution.
\newblock In {\em Proc. of SAT}, pages 147--160. Springer-Verlag, 2009.

\bibitem{tech_rep_ec_sim}
E.~Goldberg.
\newblock Equivalence checking and simulation by computing range reduction.
\newblock Technical Report arXiv:1507.02297 [cs.LO], 2015.

\bibitem{tech_rep_ec_lor}
E.~Goldberg.
\newblock Equivalence checking by logic relaxation.
\newblock Technical Report arXiv:1511.01368 [cs.LO], 2015.

\bibitem{fmcad12}
E.~Goldberg and P.~Manolios.
\newblock Quantifier elimination by dependency sequents.
\newblock In {\em FMCAD-12}, pages 34--44, 2012.

\bibitem{fmcad13}
E.~Goldberg and P.~Manolios.
\newblock Quantifier elimination via clause redundancy.
\newblock In {\em FMCAD-13}, pages 85--92, 2013.

\bibitem{tech_rep_crr}
E.~Goldberg and P.~Manolios.
\newblock Bug hunting by computing range reduction.
\newblock Technical Report arXiv:1408.7039 [cs.LO], 2014.

\bibitem{hvc-14}
E.~Goldberg and P.~Manolios.
\newblock Partial quantifier elimination.
\newblock In {\em Proc. of HVC-14}, pages 148--164. Springer-Verlag, 2014.

\bibitem{fmsd14}
E.~Goldberg and P.~Manolios.
\newblock Quantifier elimination by dependency sequents.
\newblock {\em Formal Methods in System Design}, 45(2):111--143, 2014.

\bibitem{interpolation}
K.~L. Mcmillan.
\newblock Interpolation and sat-based model checking.
\newblock In {\em CAV-03}, pages 1--13. Springer, 2003.

\end{thebibliography}
\end{document}